\documentclass{llncs}

\usepackage{amssymb}
\usepackage{amsmath}

\usepackage{graphics}
\usepackage{tikz}

\newcommand{\betstep}[1]{\rightarrow_{#1}}

\newcommand{\betredge}[1]{\betstep{#1}}

\newcommand{\wip}{w.i.p.}
\newcommand{\equilibrium}[2]{$(\alpha, k)$-equilibrium}
\newcommand{\equilibria}[2]{$(\alpha, k)$-equilibria}

\newcommand{\su}{\subseteq}

\newcommand{\m}{\mathcal} 
\newcommand{\lf}{}

\newcommand{\dg}{\text{deg}}
\newcommand{\sm}{\setminus}
\newcommand{\SW}{\text{SW}}

\newcommand{\set}[1]{\ensuremath{\{ #1 \}}}
\newcommand{\sset}[2]{\ensuremath{\set{#1 \, \mid \, #2}}}

\title{Efficient Equilibria in \newline Polymatrix Coordination Games}

\author{%
Mona Rahn\inst{1}
\and
Guido Sch\"{a}fer\inst{1,2}
}

\institute{%
Centrum Wiskunde \& Informatica (CWI), Amsterdam, The Netherlands
\and
VU University Amsterdam, The Netherlands
}

\begin{document}
\sloppy
\maketitle
\begin{abstract}
We consider polymatrix coordination games with individual preferences where every player corresponds to a node in a graph who plays with each neighbor a separate bimatrix game with non-negative symmetric payoffs.
In this paper, we study \emph{$\alpha$-approximate $k$-equilibria} of these games, i.e., outcomes where no group of at most $k$ players can deviate such that each member increases his payoff by at least a factor $\alpha$. 
We prove that for $\alpha \ge 2$ these games have the finite coalitional improvement property (and thus $\alpha$-approximate $k$-equilibria exist), while for $\alpha < 2$ this property does not hold. 
Further, we derive an almost tight bound of $2\alpha(n-1)/(k-1)$ on the price of anarchy, where $n$ is the number of players; in particular, it scales from unbounded for pure Nash equilibria ($k = 1)$ to $2\alpha$ for strong equilibria ($k = n$). 
We also settle the complexity of several problems related to the verification and existence of these equilibria.
Finally, we investigate natural means to reduce the inefficiency of Nash equilibria.
Most promisingly, we show that by fixing the strategies of $k$ players the price of anarchy can be reduced to $n/k$ (and this bound is tight). 
\end{abstract}

\section{Introduction}

\sloppy 

In this paper, we are interested in strategic games where the players are associated with the nodes of a graph and can benefit from coordinating their choices with their neighbors. More specifically, we consider \emph{polymatrix coordination games with individual preferences}: We are given an undirected graph $G = (N, E)$ on the set of players (nodes) $N := \set{1, \dots, n}$.  Every player $i \in N$ has a finite set of strategies $S_i$ to choose from and an individual preference function $q^i: S_i \rightarrow \mathbb{R}^+$. Each player $i \in N$ plays a separate bimatrix game with each of his neighbors in $N_i := \sset{j \in N}{ \set{i, j} \in E}$. In particular, every edge $\set{i, j} \in E$ is associated with a payoff function $q^{ij}: S_i \times S_j \rightarrow \mathbb{R}^+$, specifying a non-negative payoff $q^{ij}(s_i, s_j)$ that both $i$ and $j$ receive if they choose strategies $s_i$ and $s_j$, respectively. Given a joint strategy $s = (s_1, \dots, s_n)$ of all players, the overall payoff of player $i$ is defined as 
\begin{equation}\label{eq:payoff-def}
p_i(s) := q^i(s_i) + \sum_{j \in N_i} q^{ij}(s_i, s_j).
\end{equation}

These games naturally model situations in which each player has individual preferences over the available options (possibly not having access to all options) and may benefit in varying degrees from coordinating with his neighbors. For example, one might think of students deciding which language to learn, co-workers choosing which project to work on, or friends determining which mobile phone provider to use. On the other hand, these games also capture situations where players prefer to anti-coordinate, e.g., competing firms profiting equally by choosing different markets.

A special case of our games are \emph{polymatrix coordination games} (without individual preferences, i.e., $q^i = 0$ for all $i$) which have previously been investigated by Cai and Daskalakis~\cite{cai:multiplayer_minmax}. Among other results, the authors show that pure Nash equilibria are guaranteed to exist, but that finding one is PLS-complete. 
Polymatrix coordination games capture several other well-studied games among which are party affiliation games \cite{bal:advertising}, cut games \cite{christo:potential} and congestion games with positive externalities \cite{cwi:opt_cong}.

Yet another special case which will be of interest in this paper are \emph{graph coordination games}. Here every edge $\set{i, j} \in E$ is associated with a non-negative edge weight $w_{ij}$ and the payoff function $q^{ij}$ is simply defined as $q^{ij}(s_i, s_j) = w_{ij}$ if $s_i = s_j$ and $q^{ij}(s_i, s_j) = 0$ otherwise. Intuitively, in this game every player (node) $i \in N$ chooses a color $s_i$ from the set of colors $S_i$ available to him and receives a payoff equal to the total weight of all incident edges to neighbors choosing the same color. These games have recently been studied by Apt et al. \cite{us:coord} for the special case of unit edge weights. 

This paper is devoted to the study of equilibria in polymatrix coordination games with individual preferences. It is not hard to see that these games always admit pure Nash equilibria. 
However, in general these equilibria are highly inefficient. 
One of the most prominent notions to assess the inefficiency of equilibria is the \emph{price of anarchy} \cite{koutsou:worst_case}. It is defined as the ratio in social welfare of an optimal outcome and a worst-case equilibrium. Here the social welfare of a joint strategy $s$ refers to the sum of the payoffs of all players, i.e., $\SW(s) = \sum_{i \in N} p_i(s)$. 

The high inefficiency of our games even arises in the special case of graph coordination games as has recently been shown in \cite{us:coord}. To see this, fix an arbitrary graph $G = (N, E)$ with unit edge weights and suppose each player $i \in N$ can choose between a private color $c_i$ and a common color $c$. Then each player $i$ choosing his private color $c_i$ constitutes a Nash equilibrium in which every player has a payoff of zero. In contrast, if every player chooses the common color $c$ then each player $i$ obtains his maximum payoff equal to the degree of $i$. As a consequence, the price of anarchy is unbounded. 
The example demonstrates that the players might be unable to coordinate on the (obviously better) common choice because they cannot escape from a bad initial configuration by unilateral deviations. In particular, observe that the example breaks if two (or more) players can deviate simultaneously. This suggests that one should consider more refined equilibrium notions where deviations of groups of players are allowed.

In our studies, we focus on a general equilibrium notion which allows us to differentiate between both varying sizes of coalitional deviations and different degrees of player reluctance to deviate. More specifically, in this paper we consider \emph{$\alpha$-approximate $k$-equilibria} as the solution concept, i.e., outcomes that are resilient to deviations of at most $k$ players such that each member increases his payoff by at least a factor of $\alpha \ge 1$. Subsequently, we call these equilibria also \emph{\equilibria{\alpha}{k}} for short.
In light of this refined equilibrium notion, several natural questions arise and will be answered in this paper: Which are the precise values of $\alpha$ and $k$ that guarantee the existence of $(\alpha, k)$-equilibria? What is the price of anarchy of these equilibria as a function of $\alpha$ and $k$? How about the complexity of problems related to the verification and existence of such equilibria? And finally, are there efficient coordination mechanisms to reduce the price of anarchy?

\paragraph{Our contributions.}

We study \equilibria{\alpha}{k} of graph and polymatrix coordination games. 
Our main contributions are summarized below. 

\begin{enumerate}
\item \emph{Existence:}
We prove that for $\alpha \ge 2$ polymatrix coordination games have the finite $(\alpha, k)$-improvement property, i.e., every sequence of $\alpha$-improving $k$-deviations is finite (and thus results in an $(\alpha, k)$-equilibrium). 
We also exhibit an example showing that for $\alpha < 2$ this property does not hold in general. For graph coordination games we show that if the underlying graph is a tree then $(\alpha,k)$-equilibria exist for every $\alpha$ and $k$. On the other hand, if the graph is a pseudotree (i.e., a tree with exactly one cycle) the existence of \equilibria{\alpha}{k} cannot be guaranteed for every $\alpha < \varphi$ and $k \ge 2$, where $\varphi = \frac12(1+\sqrt{5})$ is the golden ratio.

\item \emph{Inefficiency:} 
We show that the price of anarchy of \equilibria{\alpha}{k} for polymatrix coordination games is at most $2\alpha(n-1)/(k-1)$. We also provide a lower bound of $2\alpha(n-1)/(k-1) + 1 - 2\alpha$. In particular, the price of anarchy drops from unbounded for pure Nash equilibria $(k = 1)$ to $2\alpha$ for strong equilibria $(k = n)$, both of which are tight bounds.

\item \emph{Complexity:} 
We settle the complexity of several problems related to the verification and existence of $(\alpha, k)$-equilibria in graph coordination games. Naturally all hardness results extend to the more general class of polymatrix coordination games with individual preferences.
A summary of our results is given in Table~\ref{tab:summary-complex}. 

\begin{table}[t]
\centering\small
\begin{tabular}{|@{\quad}l@{\quad}l@{\quad}|@{\quad}c@{\quad}|}
    \hline
	\multicolumn{2}{|@{\quad}l@{\quad}|@{\quad}}{\textbf{Problem}} & \textbf{Complexity}\\
    \hline
    \hline    
      Verification & \equilibrium{\alpha}{k} ($k$ constant) & P  \\
    	& \equilibrium{\alpha}{k} ($\alpha$ fixed) & co-NP-complete \\
    	& $\alpha$-approximate strong equilibrium  & P  \\
    \hline
    $\text{Existence}$ & $k$-equilibrium ($k \geq 2$ fixed)&  $\text{NP-complete}$  \\
    	& strong equilibrium & $\text{NP-complete}^{\dag}$ \\
    \hline
\end{tabular}
\vspace{1mm}
\caption{%
Complexity of graph coordination games. 
The parameters $\alpha$ and $k$ are assumed to be part of the input unless they are stated to be fixed. 
$^\dag$ Shown to be efficiently computable for forests. 
\label{tab:summary-complex}}
\vspace*{-5mm}
\end{table}

\item \emph{Coordination mechanisms:} 
We  investigate two natural mechanisms that a central coordinator might deploy to reduce the price of anarchy of pure Nash equilibria: (i) asymmetric sharing of the common payoffs $q^{ij}$ and (ii) strategy imposition of a limited number of players. Concerning (i), we show that there is no payoff distribution rule that reduces the price of anarchy in general. As to (ii), we prove that by (temporarily) fixing the strategies of $k$ players according to an arbitrarily given joint strategy $s$, the resulting Nash equilibrium recovers at least a fraction of $k/n$ of the social welfare $\SW(s)$ and this is best possible.
Exploiting this in combination with a $2$-approximation algorithm for the optimal social welfare problem \cite{cwi:opt_cong}, we derive an efficient algorithm to reduce the price of anarchy to at most $2n/k$ for a special class of polymatrix coordination games with individual preferences.

\end{enumerate}

\paragraph{Related work.}
Apt et al.~\cite{us:coord} study $k$-equilibria in graph coordination games with unit edge weights, which constitute a special case of our games. They identify several graph structural properties that ensure the existence of such equilibria. Interestingly, most of these results do not carry over to our weighted graph coordination games, therefore demanding for the new approach of considering approximate equilibria.

Many of the mentioned games have been studied from a computational complexity point of view. In particular, Cai and Daskalakis \cite{cai:multiplayer_minmax} show that the problem of finding a pure Nash equilibrium in a polymatrix coordination game is PLS-complete. Further, they show that finding a mixed Nash equilibrium is in PPAD $\cap$ PLS. 
While this suggests that the latter problem is unlikely to be hard, it is not known whether it is in P. It is easy to see that these results also carry over to our polymatrix coordination games with individual preferences.\footnote{In \cite{cai:multiplayer_minmax} the bimatrix games on the edges may have negative payoffs and this is exploited in the PLS-completeness proof. However, we can accommodate this in our model by adding a sufficiently large constant to each payoff.}

For the special case of party affiliation games efficient algorithms to compute an approximate Nash equilibrium are known \cite{bhal:party_affiliation,cara:approx_eq}. The current best approximation guarantee is $3+\varepsilon$, where $\varepsilon > 0$, due to Caragiannis, Fanelli and Gavin \cite{cara:approx_eq}. The algorithm crucially exploits that party affiliation games admit an exact potential whose relative gap (called \emph{stretch}) between any two Nash equilibria is bounded by 2. The latter property is not satisfied in our games, even for graph coordination games (as the example outlined in the Introduction shows).

A class of games that is closely related to our graph coordination games are \emph{additively separable hedonic games} \cite{bogo:hedonic}. As in our games, the players are embedded in a weighted graph. Every player chooses a coalition and receives as payoff the total weight of all edges to neighbors in the same coalition. These games were originally studied in a cooperative game theory setting. More recently, researchers also address computational issues of these games (see, e.g.,  \cite{aziz:stable_ashg}).
It is important to note that in hedonic games every player can choose every coalition, while in our graph coordination games players may only have limited options.

Anshelevich and Sekar \cite{anshe:social_coord} study coordination games with individual preferences where the players are nodes in a graph and profit from neighbors choosing the same color. However, in their setting the edge weight between two neighbors can be distributed asymmetrically and all players are assumed to have the same strategy set. 
Among other results, they give an algorithm to compute a $(2,n)$-equilibrium and show how to efficiently compute an approximate equilibrium that is not too far from the optimal social welfare.

Concerning the social welfare optimization problem, a $2$-approximation algorithm is given in \cite{cwi:opt_cong} for the special case of polymatrix coordination games with individual preferences where the bimatrix game of each edge has positive entries only on the diagonal.

\paragraph{Our techniques.}

Most of our existence results use a generalized potential function argument for coalitional deviations. In our proof of the upper bound on the inefficiency of $(\alpha, k)$-equilibria we first argue locally for a fixed coalition of players and then use a sandwich bound in combination with a counting argument to derive the upper bound. Most of our lower bounds and hardness results follow by exploiting specific properties and deep structural insights of graph coordination games with edge weights.

It is worth mentioning that our algorithm to compute a strong equilibrium for graph coordination games on trees reveals a surprising connection to a sequential-move version of the game. In particular, we show that if we fix an arbitrary root of the tree and consider the induced sequential-move game then every subgame perfect equilibrium corresponds to a strong equilibrium of the original game. As a consequence, strong equilibria exist and can be computed efficiently. Further, this in combination with our strong price of anarchy bound shows that the \emph{sequential price of anarchy} \cite{LST12} for these induced games is at most $2$, which is a significant improvement over the unbounded price of anarchy for the strategic-form version of the game. This result is of independent interest.

We also note that the $k/n$ bound on the social welfare which is guaranteed by our strategy imposition algorithm is proven via a \emph{smoothness argument} \cite{rough:rpoa}. Besides some other consequences, this implies that our bound also holds for more permissive solution concepts such as correlated and coarse correlated equilibria (see \cite{rough:rpoa} for more details).

\section{Preliminaries}
\label{sec:preliminaries}

Let $\mathcal{G} = (G, (S_i)_{i \in N}, (q^i)_{i \in N}, (q^{ij})_{\set{i,j}\in E})$ be a polymatrix coordination game with individual preferences (\wip) where $G = (N, E)$ is the underlying graph.
Recall that we identify the player set $N$ with $\set{1, \dots, n}$. We first introduce some standard game-theoretic concepts. 

We call a subset $K := \{i_1, \ldots, i_k\} \subseteq N$ of players a \emph{coalition of size $k$}. We define the set of joint strategies of players in $K$ as $S_K := S_{i_1} \times \cdots \times S_{i_{k}}$ and use $S := S_N$ to refer to the set of joint strategies of all players. Given a joint strategy $s \in S$, we use $s_K$ to refer to $(s_{i_1}, \ldots, s_{i_{k}})$ and $s_{-K}$ to refer to $(s_i)_{i \notin K}$. By slightly abusing notation,  we also write $(s_K, s_{-K})$ instead of $s$. If there is a strategy $x$ such that $s_i = x$ for every player $i \in K$, we also write $s = (x_K, s_{-K})$.

Given a joint strategy $s$ and a coalition $K$, we say that $s' = (s'_K, s_{-K})$ is a \emph{deviation of coalition $K$ from $s$} if $s'_i \neq s_i$ for every player $i \in K$; we also denote this by $s \betredge{K} s'$. If we constrain to deviations of coalitions of size at most $k$, we call such deviations \emph{$k$-deviations}. 
We call a deviation \emph{$\alpha$-improving} if every player in the coalition improves his payoff by at least a factor of $\alpha \ge 1$, i.e., for every $i \in K$, $p_i(s') > \alpha p_i(s)$; we also call such deviations \emph{$(\alpha, k)$-improving}. 
We omit the explicit mentioning of the parameters if $\alpha = 1$ or $k = 1$. 
A joint strategy $s$ is an \emph{$\alpha$-approximate $k$-equilibrium} (also called \emph{$(\alpha,k)$-equilibrium} for short) if there is no $(\alpha, k)$-improving deviation from $s$. If $k = 1$ or $k = n$ then we also refer to the respective equilibrium notion as \emph{$\alpha$-approximate Nash equilibrium} and \emph{$\alpha$-approximate strong equilibrium} \cite{aum:strong}.

We say that a finite strategic game has the \emph{finite $(\alpha, k)$-improvement property} (or \emph{$(\alpha, k)$-FIP} for short) if every sequence of $(\alpha, k)$-improving deviations is finite. This notion generalizes the finite improvement property introduced by Monderer \cite{mon:potential} for $\alpha =  k = 1$. A function $\Phi: S \to \mathbb R$ is called an \emph{$(\alpha, k)$-generalized potential} if for every joint strategy $s$, for every $(\alpha, k)$-improving deviation $s' := (s'_{K}, s_{-K})$ from $s$ it holds that $\Phi(s') > \Phi(s)$. It is not hard to see that if a finite game admits an $(\alpha, k)$-generalized potential then it has the $(\alpha, k)$-FIP. 

The \emph{social welfare} of a joint strategy $s$ is defined as $\SW(s):=\sum_{i \in N} p_i(s)$. For $K \subseteq N$, we define $\SW_K(s) := \sum_{i \in K} p_i(s)$. A joint strategy $s^*$ of maximum social welfare is called a \emph{social optimum}. Given a finite game that has an \equilibrium{\alpha}{k}, its \emph{$(\alpha, k)$-price of anarchy (POA)} is the ratio $\SW(s^*)/\SW(s)$, where $s^*$ is a social optimum and $s$ is an \equilibrium{\alpha}{k} of smallest social welfare. In the case of division by zero, we interpret the outcome as $\infty$. Note that if $\alpha' \geq \alpha$ and $k' \leq k$, then every $(\alpha, k$)-equilibrium is an $(\alpha', k')$-equilibrium. Hence the $(\alpha, k)$-PoA lower bounds the $(\alpha', k')$-PoA. 

Due to lack of space, several proofs are omitted from the main part of this paper and can be found in the appendix.

\section{Existence}
\label{sec:existence}

We first give a characterization of the values $\alpha$ and $k$ for which our polymatrix coordination games with individual preferences have the $(\alpha, k)$-FIP. 

\begin{theorem}
\label{thm:ex_ne}
Let $\m G$ be a polymatrix coordination game \wip\ Then:
\begin{enumerate}
\item $\m G$ has the $(\alpha, 1)$-FIP for every $\alpha$.
\item $\m G$ has the $(\alpha, k)$-FIP for every $\alpha \ge 2$ and for every $k$.
\end{enumerate}
\end{theorem}
\begin{proof}
Observe that every $\alpha$-improving deviation is also $\alpha'$-improving 
for $\alpha \ge \alpha'$. It is thus sufficient to prove the claims above for $\alpha = 1$ and $\alpha = 2$, respectively.  

In order to prove the first claim we prove that the game admits an exact potential and thus has the FIP; details are given in Appendix~\ref{sec:app:1}. 

We next prove the second claim for $\alpha = 2$ by showing that $\Phi(s) := \SW(s)$ is a $(2, k)$-generalized potential. Given a joint strategy $s$ and two sets $K, K' \subseteq N$, define 
$$
Q_s(K, K') := \sum_{i \in K,\, j \in N_i \cap K'} q^{ij}(s)
\quad\text{and}\quad 
Q_s(K) := \sum_{i \in K} q^i(s).
$$
Consider a $(2, k)$-improving deviation $s' = (s'_K, s_{-K})$ from $s$. Let $\bar K$ be the complement of $K$. We have $\SW_K(s) = Q_s(K, K) + Q_s(K, \bar K) + Q_s(K)$.
Note that $\SW_K(s') > 2\SW_K(s)$ because the deviation is $2$-improving. 
Thus,
\begin{equation}
\label{eqn:Qs}
Q_{s'}(K, K) + Q_{s'}(K, \bar K) + Q_{s'}(K)> 2 \big(Q_{s}(K, K) + Q_s (K, \bar K) + Q_s(K)\big).
\end{equation}
The social welfare of $s$ can be written as
\[
\SW(s) = Q_s(K, K) + 2Q_s(K, \bar K) + Q_s(\bar K , \bar K) + Q_s(K) + Q_s(\bar K).
\]
Note that $Q_s(\bar K , \bar K) = Q_{s'} (\bar K , \bar K)$ and $Q_s(\bar K) = Q_{s'}(\bar K)$. Using \eqref{eqn:Qs}, we obtain 
\begin{align*}
\Phi(s') - \Phi(s) 
 & = Q_{s'}(K , K) + 2Q_{s'}(K , \bar K) + Q_{s'}(K) \\
 & \quad - Q_s(K, K) - 2Q_s(K , \bar K) - Q_s(K) \\
 & > Q_s(K , K) + Q_{s'}(K , \bar K) + Q_{s}(K) \ge 0.
\end{align*}
Thus $\Phi(s)$ is a $(2, k)$-generalized potential which concludes the proof. 
\qed
\end{proof}

The next theorem shows that in general our polymatrix coordination games do not have the $(\alpha, k)$-FIP for $\alpha < 2$.  

\begin{theorem}\label{thm:no-FIP}
For all $\alpha<2$ there is a polymatrix coordination game $\m G$ that has a cycle of $(\alpha, n-1)$-improving deviations.
\end{theorem}

We derive some more refined insights for the special case of graph coordination games. 

\begin{theorem}\label{thm:graph-coord-exist}
The following holds for graph coordination games:
\label{thm:ex_tree}
\label{thm:graph_without_2eq}
\begin{enumerate}
\item Let $\m G$ be a graph coordination game on a tree. Then $\m G$ has a strong equilibrium.
\item There is a graph coordination game $\m G$ on a graph with one cycle such that no \equilibrium{\alpha}{k} exists for every $\alpha < \varphi$ and $k \ge 2$, where $\varphi := \frac{1}{2}(1+\sqrt 5) \approx 1.62$ is the golden ratio. 
\end{enumerate}
\end{theorem}

Note that Theorem~\ref{thm:graph-coord-exist} shows that for $k \ge 2$ a $k$-equilibrium may not exist. In contrast, Nash equilibria always exist by Theorem~\ref{thm:ex_ne}. Further, the graph used to show the second claim is a pseudoforest\footnote{A graph is a \emph{pseudoforest} if each of its connected components has at most one cycle.}. 
For graph coordination games with unit edge weights, this guarantees the existence of a strong equilibrium \cite{us:coord}. 

\section{Inefficiency}
\label{sec:inefficiency}

We analyze the price of anarchy of our polymatrix coordination games. The upper bound in the special case of $(\alpha,k) = (1,n)$ follows from a result in \cite{bach:strong_poa}. 

\begin{theorem}
\label{thm:approx_poa}
The $(\alpha, k)$-price of anarchy in polymatrix coordination games \wip\ is between $2\alpha(n-1)/(k-1) + 1 - 2\alpha$ and $2\alpha(n-1)/(k-1)$. 
The upper bound of $2 \alpha$ is tight for $\alpha$-approximate strong equilibria. 
\end{theorem}
\begin{proof}[upper bound]
Let $s$ be an $(\alpha, k)$-equilibrium (which we assume to exist) and let $s^*$ be a social optimum. Fix an arbitrary coalition $K = \set{i_1, \dots, i_k}$ of size $k$. Then there is a player $i \in K$ such that $p_i(s^*_{K}, s_{-K}) \leq \alpha p_i(s)$. Denote by $p_i^K(s^*) := q^i(s^*) + \sum_{j \in N_i \cap K} q^{ij}(s^*)$ the total payoff that $i$ gets from players in $K$ under $s^*$ (including himself). Because all payoffs are non-negative, we have 
\begin{align}
p_i^K(s^*) 
& \leq q^i(s^*_i) + \sum_{j \in N_i \cap K} q^{ij} (s^*_i, s^*_j) + \sum_{j \in N_i \cap \bar K} q^{ij}(s^*_i, s_j) = p_i(s^*_K, s_{-K}). \label{eq:ineff}
\end{align}
Thus, $p_i^K(s^*) \leq \alpha p_i(s)$. Rename the nodes in $K$ such that $i_k = i$ and repeat the arguments above with $K \setminus \set{i_k}$ instead of $K$. Continuing this way, we obtain that for every player $i_x \in K$, $p_{i_x}^{\set{i_1, \dots, i_x}}(s^*) \le \alpha p_{i_x}(s)$. 

We thus have 
\begin{align*}
\sum_{i \in K} \Big( q^i(s^*) + \frac{1}{2} \sum_{j \in N_i \cap K} q^{ij}(s^*) \Big)
= \sum_{x = 1}^k p_{i_x}^{\set{i_1, \dots, i_x}}(s^*) 
\le \alpha \sum_{i \in K} p_{i}(s) 
\end{align*}

Summing over all coalitions $K$ of size $k$, we obtain
\begin{align}
\label{eqn:SW_a2}
\sum_{K: |K| = k} \Big(\sum_{i \in K} \Big( q^i(s^*) + \frac{1}{2} \sum_{j \in N_i \cap K} q^{ij}(s^*) \Big) \Big)
\le \alpha \sum_{K: |K| = k} \sum_{i \in K} p_{i}(s).
\end{align}

Consider the right-hand side of \eqref{eqn:SW_a2}. Note that every player $i \in N$ occurs in $n-1 \choose k-1$ many coalitions of size $k$ because we can choose $k-1$ out of $n-1$ remaining players to form a coalition of size $k$ containing $i$. Thus 
\begin{align}
\sum_{K: |K| = k} \sum_{i \in K} p_{i}(s) &
= {n-1 \choose k-1} \sum_{i \in N} p_i(s) = {n-1 \choose k-1} \SW(s). \label{eq:ineff1}
\end{align}

Similarly, the first term of the left-hand side of \eqref{eqn:SW_a2} yields
\[
\sum_{K: |K| = k} \sum_{i \in K} q^i(s^*) = {n-1 \choose k-1} \sum_{i \in N} q^i(s^*) \geq \frac{1}{2} {n-2 \choose k-2} \sum_{i \in N} q^i(s^*).
\]

Now, consider the second term of the left-hand side of \eqref{eqn:SW_a2}. Every pair $(i, j)$ with $i \in N$ and $j \in N_i$ occurs in ${n-2 \choose k-2}$ many coalitions of size $k$ because we can choose $k-2$ out of $n-2$ remaining players to complete a coalition of size $k$ containing both $i$ and $j$. Thus for the left-hand side of \eqref{eqn:SW_a2} we obtain
\begin{align}
\sum_{K: |K| = k} & \Big(\sum_{i \in K} \Big( q^i(s^*) + \frac{1}{2} \sum_{j \in N_i \cap K} q^{ij}(s^*) \Big) \Big) \notag \\
&\geq \frac{1}{2} {n-2 \choose k-2} \Big ( \sum_{i \in N} q^i(s^*) + \sum_{i \in N}\sum_{j \in N_i} q^{ij}(s^*) \Big) = \frac{1}{2} {n-2 \choose k-2} \SW(s^*). \label{eq:ineff2}
\end{align}

Combining \eqref{eq:ineff1} and \eqref{eq:ineff2} with inequality \eqref{eqn:SW_a2}, we obtain that the $(\alpha, k)$-price of anarchy is at most $2 \alpha {n-1 \choose k-1}/{n-2 \choose k-2} = 2\alpha \frac{n-1}{k-1}$. 
\qed
\end{proof}

\section{Complexity}
\label{sec:complexity}

In this section, we study the complexity of various computational problems on graph coordination games.

\begin{theorem}\label{thm:verification}
Let $\m G$ be a graph coordination game. Given a joint strategy $s$, the problem of deciding whether $s$ is an \equilibrium{\alpha}{k}
\begin{enumerate}
\item is in $P$, if $k = O(1)$ or $k = n$;
\item is co-NP-complete for every fixed $\alpha$.
\end{enumerate}
\end{theorem}

\begin{proof}[sketch]
We sketch the proof of the first claim for $k = n$. A crucial insight is that if there is an $\alpha$-improving deviation from $s$ then there is one which is \emph{simple}, i.e., $s' = (s'_K, s_{-K})$ where the subgraph $G[K]$ induced by $K$ is connected and all nodes in $K$ deviate to the same color $s'_K = x$ for some $x$ (see Lemma~\ref{lem:simple_ext}). 

Fix some color $x$ and let $G_x := (N_x, E_x)$ be the subgraph of $G$ induced by the set of nodes $N_x$ that can choose color $x$ but do not do so in $s$. For each $u \in N_x$ define $d_u :=  \alpha p_u(s) - w(\sset{\set{u,v} \in E}{s_j = x}) - q^u(x)$. 
Now, a deviation of a coalition $K \su N_x$ to $(x_K, s_{-K})$ is $\alpha$-improving if and only if for every node $u \in K$ the total weight of all incident edges in the induced subgraph $G_x[K]$ is larger than $d_u$. We prove that an inclusionwise maximal $K \su N_x$ satisfying this property can be found in polynomial time. This way we can verify for every color $x$ whether an $\alpha$-improving deviation exists. 
\qed
\end{proof}

Deciding whether a graph coordination game admits a $k$-equilibrium is hard for every $k \ge 2$. Note that for unit edge weights $2$-equilibria are guaranteed to exist and can be found efficiently, as shown in \cite{us:coord}.

\begin{theorem}\label{thm:decide-existence}\label{thm:computation_2eq}
Let $\m G$ be a graph coordination game. Then the problem of deciding whether there is a $k$-equilibrium is NP-complete for every fixed $k \ge 2$. 
\end{theorem}
\begin{proof}[$k = 2$]
We give a reduction from \textsc{minimum maximal matching} which is known to be NP-complete  \cite{yanna:edge_dominating}: Given a graph $G = (V, E)$ and a number $l$, does there exist an inclusionwise maximal matching of size at most $l$? 

Let $(G, l)$ be an instance of this problem with $G = (V, E)$ and $n = |V|$. 
We add  $n-2l$ gadgets $H_1, \ldots, H_{n-2l}$ to $G$, where an illustration of gadget $H_i$ is given in Figure~\ref{fig:gadget}. The dashed edge from $v_0^i$ to $G$ indicates that $v_0^i$ is connected to all vertices in $G$ and each of these edges has weight 3. 
We assign to each node $v \in V$ the color set $S_v = \{x^i_v \mid i = 1, \ldots, n-2l\} \cup \{y_e \mid e = \{v,w\} \in E\}$, i.e., $v$ can either choose a `gadget color' $x^i_v$ or a color corresponding to some adjacent edge in $E$. Every edge in $E$ has weight $4$. 
Note that for all joint strategies of nodes in $V$ the set of unicolor edges in $E$ constitutes a matching. 
The idea is that in every 2-equilibrium $n-2l$ nodes in $V$ are needed to `stabilize' the gadgets and the $2l$ remaining nodes in $V$ form a maximal matching.

\begin{figure}[t]
\begin{center}
  \begin{tikzpicture}[scale = 1.5] 
  	\draw [densely dotted] (-.5,.35) ellipse (2.5cm and 1.2cm);
  	\node (G) at (0,2) {$G$};
    \node [label=left:{\lf $\{a,c\} \cup \{x_v^i \mid v \in V\}$}] (v_0) at (0, 1) {$v^i_0$};
    \node [label=below:{\lf $\{a,b\}$}] (v_1) at (-1, 0) {$v^i_1$};
    \node [label=below:{\lf $\{b,c\}$}] (v_2) at (1, 0) {$v^i_2$};
    \node [label=below:{\lf $\{b\}$}] (u) at (-2, 0) {$u^i$}; 
	\node (H) at (2.25,.15) {$H_i$};
	
	\draw[dashed] (G) -- node[right,pos=.25]{$3$} (v_0);
    \draw (v_1) -- node[auto]{$4$} (v_0);                                                       
    \draw (v_0) -- node[auto]{$3$} (v_2);                                                       
    \draw (v_1) -- node[below]{$2$} (v_2);                                                       
    \draw (u) -- node[auto]{$3$} (v_1);                                                       
\end{tikzpicture}
\caption{The gadget $H_i$.}
\label{fig:gadget}
\end{center}
\vspace*{-5mm}
\end{figure}
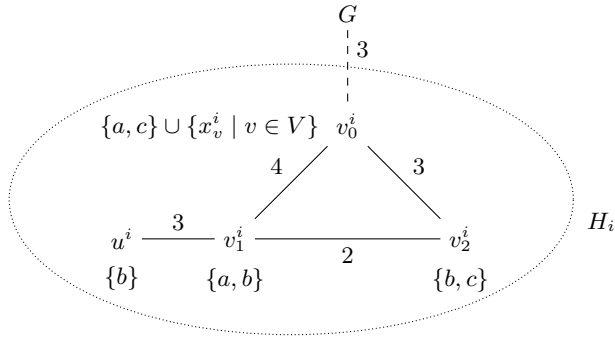

Assume that $G$ has a maximal matching $M \su E$ with $|M| \le l$. We construct a 2-equilibrium $s$. For every matched node $v \in V(M)$, choose the color corresponding to the adjacent matching edge. On the unmatched nodes in $V$ and nodes of the form $v^i_0$, we assign colors in such a way that every gadget has one outgoing edge (indicated by the dashed edge) that is unicolor. This is possible because there are at least $n-2l$ unmatched nodes in $V$. If there are uncolored nodes in $V$ left, assign arbitrary colors to them. Finally, let $v_1^i$ and $v_2^i$ choose color $b$ for every $i$. 
We claim that $s$ is a 2-equilibrium: The matched nodes obtain a payoff of 4, which is the maximal payoff nodes in $V$ can get; so they are not part of any improving deviation. Let $v \in V$ be unmatched. Then $v$ cannot deviate together with another unmatched node to increase the payoff because $M$ is maximal. Further, all gadget nodes are `taken': every $v_0^i$ has a payoff of $3$, which a joint deviation with $v$ cannot increase. This implies that $v$ cannot be part of any improving deviation. Lastly, it is easy to see that pairs of gadget nodes cannot profitably deviate. This proves that $s$ is a 2-equilibrium.

Conversely, assume that a joint strategy $s$ is a 2-equilibrium. Let $M$ consist of the unicolor edges in $G$. By the choice of the color assignment, $M$ is a matching. $M$ is maximal because if there were two unmatched adjacent nodes, then they could form a profitable deviating coalition. It remains to show that $|M| \leq l$. It is not hard to see that if there is a gadget without an outgoing unicolor edge, then there is a $2$-improving deviation in $H_i$. So at least $n-2l$ nodes choose gadget colors, implying that $|V(M)| \leq 2l$ and thus $|M| \leq l.$
\qed
\end{proof}

On the positive side, we can compute a strong equilibrium in polynomial time if the underlying graph is a tree. 

\begin{theorem}
\label{thm:SE_forest} 
Let $\m G$ be a graph coordination game on a tree. Then there is a polynomial-time algorithm to compute a strong equilibrium. 
\end{theorem}
\begin{proof}[sketch]
The idea is as follows: We fix an arbitrary root $r$ of the tree and consider the induced sequential-move game. This game has a subgame perfect equilibrium $s$ which can be computed in polynomial time by backwards induction. Let $\bar s$ be the corresponding joint strategy of $\m G$ if every player plays his best response according to $s$. We can prove that $\bar s$ is a strong equilibrium of $\m G$. 
\end{proof}

Example \ref{ex:polymatrix_tree} in the appendix shows that, unfortunately, the idea above does not extend to polymatrix coordination games.

\section{Coordination mechanisms}
\label{sec:coord-mech}

In this section, we investigate means that a central designer could use to reduce the inefficiency of Nash equilibria. 

In our games the common payoff $q^{ij}$ of the bimatrix game on edge $\set{i, j}\in E$ is distributed equally to both $i$ and $j$. An idea that arises is to use different \emph{payoff sharing rules} to reduce the inefficiency. 
Unfortunately, it turns out that the price of anarchy remains unbounded no matter which payoff sharing rule is used (details are given in the appendix). 

We therefore consider another natural approach. Suppose the central designer can impose strategies on a subset of the players to reduce the inefficiency. Let $\m G$ be a polymatrix coordination game \wip\ Further, let $K \su N$ be a subset of the players and fix a joint strategy $f_K \in S_K$ for players in $K$. We define $\m G[f_K]$ as the game with players from $N \setminus K$ that arises from $\m G$ if we fix the strategies of all players in $K$ according to $f_K$. We say that $f_K$ \emph{guarantees} social welfare $z$ if $\SW(f_K, s_{-K}) \ge z$ for all Nash equilibria $s_{-K}$ of $\m G[f_K]$. We also call $f_K$ a \emph{joint strategy of size $|K|$}. 

Suppose that $f_K$ guarantees social welfare $z$. Then once all players in $\m G[f_K]$ have reached a Nash equilibrium we can release all players in $K$ and let them play their best responses too. By Theorem~\ref{thm:ex_ne}, the social welfare can only increase subsequently. As a result, the final Nash equilibrium has social welfare at least $z$. So we can view $f_K$ as a `temporary advice' for the players in $K$. A similar idea has been put forward in \cite{bal:advertising}. 

We first show that determining the minimum number of players to guarantee a certain social welfare is hard, even for graph coordination games. 

\begin{theorem}\label{thm:strat-imp}
Let $\m G$ be a graph coordination game. Given a joint strategy strategy $s$, the problem of finding a minimal $k$ such that there is a joint strategy $f_K$ of size $k$ that guarantees social welfare $\SW(s)$ is NP-hard. The claim also holds if $f_K$ is restricted to be $s_K$.
\end{theorem}

In light of the above hardness results, we resort to approximation algorithms. 

\begin{theorem}
\label{thm:forcing_sw}
Let $\m G$ be a polymatrix coordination game \wip\ Given a joint strategy $s$ and a number $k$, we can find in polynomial time a coalition $K$ of size $k$ such that $s_K$ guarantees welfare $\frac{k}{n} \SW(s)$ and this is tight. 
\end{theorem}

Using the $2$-approximation algorithm for the social welfare optimization problem in \cite{cwi:opt_cong}, we obtain the following corollary.

\begin{corollary}
Let $\m G$ be a polymatrix coordination game \wip\ where the bimatrix game of every edge has positive entries on the diagonal only . Given a number $k$, we can compute a joint strategy $f_K$ of size $K$ that guarantees a $\frac{k}{2n}$ fraction of the optimal social welfare.
\end{corollary}

\bibliographystyle{abbrv}
\bibliography{monabib}

\newpage
\appendix

\section{Missing proofs of Section~\ref{sec:existence}}
\label{sec:app:1}

\begin{proof}[Theorem~\ref{thm:ex_ne}, Claim 1]
We show that $\m G$ admits an exact potential. We can decompose $\m G$ into a game $\m G_1$ in which player $i$'s payoff is $q^i$ and a polymatrix coordination game $\m G_2$ in which player $i$ receives payoff $\sum_{j \in N_i} q^{ij}$. Clearly, $\sum_{i \in N} q^i(s)$ is an exact potential for $\m G_1$. Further, it is known that half the social welfare is an exact potential for $\m G_2$; see \cite{cai:multiplayer_minmax}. Thus
$$
\Phi(s) := \sum_{i \in N} q^i(s) + \frac{1}{2} \sum_{i \in N} \sum_{j \in N_i} q^{ij}(s)
= \sum_{i \in N} q^i(s) + \sum_{\set{i, j} \in E} q^{ij}(s)
$$
is an exact potential for $\m G$. 
\qed
\end{proof}

\begin{proof}[Theorem~\ref{thm:no-FIP}]
Denote by $\oplus$ (resp. $\ominus$) the addition (resp. subtraction) modulo $n-1$.
Consider the following polymatrix coordination game on $n$ players. For convenience we assume here that the set of players is given as $N = \set{1, \dots, n-1}$. The strategy set of player $i$ is $S_i = N$. One way to think of this is that every player may choose another player to support (including himself). For two players $i \neq j$, we define the bimatrix game $q^{ij}$ is as follows. If $i$ and $j$ both support $j$, they get $2^{i \ominus j}$. So both get $2^{n-2}$ if `$i$ is $j$'s left neighbor', $2^{n-3}$ if `$i$ is $j$'s second left neighbor', and so on. Similarly, if $i$ and $j$ both choose $i$, they get $2^{j \ominus i}$. In all other cases, $q^{ij} (s_i, s_j) = 0.$ Figure \ref{fig:payoffs} depicts the bimatrix payoffs if all players support player $n-1.$

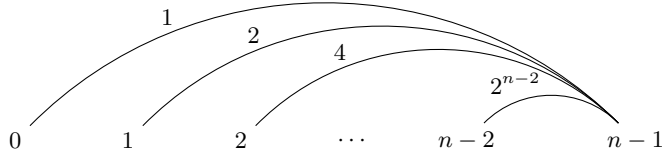
\begin{figure}[t]
\begin{center}
\begin{tikzpicture}[scale=1.5]
\node (0) at (0, 0) {$0$};
\node (1) at (1, 0) {$1$};
\node (2) at (2, 0) {$2$};
\node () at (3, 0) {$\cdots$};
\node (n-2) at (4,0) {$n-2$};
\node (n) at (5.5, 0) {$n-1$};
\path (0) edge [bend left=45] node[above, pos=.25, black] {$1$} (n); 
\path  (1) edge [bend left=45] node[above, pos=.25, black] {$2$} (n);
\path  (2) edge [bend left=45] node[above, pos=.25, black] {$4$} (n);
\path  (n-2) edge [bend left=45] node[above, pos=.25, black] {$2^{n-2}$} (n);
\end{tikzpicture}
\end{center}
\caption{The edges are labeled with their respective bimatrix payoffs when every player supports player $n-1.$
\label{fig:payoffs}}
\end{figure}

For $i = 0, \ldots, n-1$, let $s^i$ be the strategy profile in which player $i\oplus 1$ supports himself and all other players support player $i$. Let $K_i = N \sm \{i\}.$ We claim that
\[
 s^{n-1} \: \betredge{K_{n-1}} s^{n-2} \: \betredge{K_{n-2}} \:  \:\ldots \:\betredge{K_{2}} \: s^{1} \: \betredge{K_{1}} \:  s^0 = s^{n-1}
\]
is a cycle of $\alpha$-improving deviations, where $\alpha = 2 - \frac{1}{2^{n-3}}$.

By symmetry, it is enough to show that $s^{n-1} \betredge{K_{n-1}} s^{n-2}$ is $\alpha$-improving. Note that indeed, $s^{n-1}$ and $s^{n-2}$ only differ in the strategies of $K_{n-1} = N \sm \{n-1\}.$ The payoffs in $s^{n-1}$ are as follows. Player $0$ gets 0 because he is the only one supporting himself. Players $0 < i < n-1$ get $q^{i,n-1}(n-1,n-1) = 2^{i}$ from supporting $n-1$. Finally, player $n-1$ gets $\sum_{i=1}^{n-2} 2^i = 2^{n-1}-2$. 

In $s_{n-1}$, every player $i \neq n-1$ gets the utility that $i+1$ previously got. For all players in $K_{n-1}$ except $n-2$, this is at least twice her previous utility. Player $n-2$ improves her utility by $(2^{n-1}-2) / 2^{n-2} = \alpha.$ So $s^{n-1} \betredge{K_{n-1}} s^{n-2}$ is indeed an $\alpha$-improving joint deviation.

Letting $n$ go to infinity, this proves the theorem for arbitrary $\alpha<2$.
\qed
\end{proof}

\begin{proof}[Theorem~\ref{thm:graph_without_2eq}]
The first claim follows from Theorem~\ref{thm:SE_forest}.

For the second claim consider the graph coordination game depicted in Figure \ref{fig:alpha_approx}. We argue that there is no $(\alpha, 2)$-equilibrium for every $\alpha < \varphi$. Suppose there is one and call it $s.$ Then by definition of the color sets, only one of the edges in the triangle can be unicolor. By symmetry, we assume it is $\{v_0, v_1\}.$ We can assume $s_{v_2}=z$ (if $s_{v_2}=y$, $v_1$ can profitably deviate). Then the joint deviation of $K = \{v_1, v_2\}$ to $y$ yields a payoff increase of $\varphi > \alpha$ (for $v_2$) and $\frac{1+\varphi}{\varphi} = \varphi > \alpha$ (for $v_1$). So $s$ is not an $\alpha$-approximate 2-equilibrium.
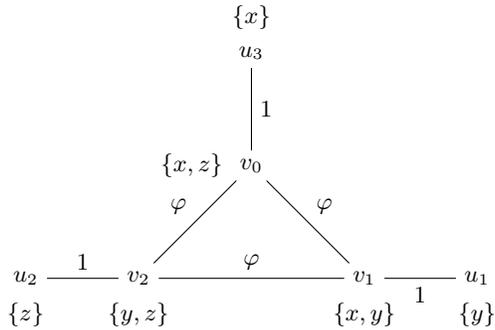
\begin{figure}[t]
\begin{center}
\begin{tikzpicture}[scale=1.5]
    \node [label=left:{\lf $\{x,z\} $}] (v_0) at (0, 1) {$v_0$};
    \node [label=below:{\lf $\{y,z\}$}] (v_2) at (-1, 0) {$v_2$};
    \node [label=below:{\lf $\{x,y\}$}] (v_1) at (1, 0) {$v_1$};
    \node [label=above:{\lf $\{x\}$}] (u_0) at (0, 2) {$u_3$}; 
    \node [label=below:{\lf $\{y\}$}] (u_1) at (2, 0) {$u_1$}; 
    \node [label=below:{\lf $\{z\}$}] (u_2) at (-2, 0) {$u_2$}; 
	
    \draw (v_2) -- node[auto]{$\varphi$} (v_0);                                                       
    \draw (v_0) -- node[auto]{$\varphi$} (v_1);                                                       
    \draw (v_2) -- node[auto]{$\varphi$} (v_1);                                                       
    \draw (u_0) -- node[auto]{$1$} (v_0);                                                       
    \draw (u_1) -- node[auto]{$1$} (v_1);
    \draw (u_2) -- node[auto]{$1$} (v_2);                                                       
\end{tikzpicture}
\caption{A graph coordination game with no $\alpha$-approximate 2-equilibrium for every $\alpha < \varphi$.
\label{fig:alpha_approx}}
\end{center}
\end{figure}
\qed
\end{proof}

\section{Missing proofs of Section~\ref{sec:inefficiency}}
\label{sec:app:2}

\begin{proof}[Theorem~\ref{thm:approx_poa}, lower bounds]
We first prove that the $\alpha$-approximate strong price of anarchy is at least $2\alpha$ for the class of graph coordination games.

Consider the graph coordination game depicted in Figure~\ref{fig:alpha_approx_spoa}. 
\begin{figure}[h]
\begin{center}
\begin{tikzpicture}[scale=1.5]
	\node [label=above:{\lf $\{a\}$}] (v_1) at (1, 0) {$v_1$};
	\node [label=above:{\lf $\{a,b\}$}] (v_2) at (2, 0) {$v_2$};
	\node [label=above:{\lf $\{c,b\}$}] (v_3) at (3, 0) {$v_3$};
	\node [label=above:{\lf $\{c\}$}] (v_4) at (4, 0) {$v_4$};
	\draw (v_1) -- node[auto]{$\alpha$} (v_2);                                                       
	\draw (v_2) -- node[auto]{$1$} (v_3);                                                       
	\draw (v_3) -- node[auto]{$\alpha$} (v_4);                                                       
\end{tikzpicture}
\end{center}
\caption{\label{fig:alpha_approx_spoa}}
\end{figure}

The strategy profile in which $v_2$ and $v_3$ choose $b$ is an $\alpha$-approximate strong equilibrium with social welfare $2$, while the optimum (in which they choose $a$ and $c$, resp.) achieves $4\alpha.$

Next, we show that for graph coordination games and all $n \geq k > 1$ the $\alpha$-approximate $k$-price of anarchy is at least 
\[
2\alpha \Big(\frac{n-1}{k-1}-1 \Big) + 1
\]
The $\alpha$-approximate price of anarchy is $\infty$ \cite{us:coord}.

Fix $n$ and $1\leq k \leq n$. Let $N$ consist of two sets $V_1$ and $V_2$ of size $k$ and $n-k$, respectively, and define
\[
E = \{\{u,v\} \mid u \in V_1, v \in V_1 \cup V_2\}.
\]
We give the edges in $E[V_1]$ weight 1 and the edges between $V_1$ and $V_2$ weight $\alpha.$

Fix three colors $a,b$ and $c$. For $v \in V_1$, let $S_v = \{a,c\}$. For $v \in V_2$, let $S_v = \{b,c\}$. Then the color assignment $\sigma$ in which each player chooses the common color $c$ is a social optimum.
The social welfare is
\[
\SW(\sigma) = k(k-1) + 2 \alpha k(n-k) = k \big(k-1 + 2 \alpha ((n-1) - (k-1))\big).
\]

Next we show that the color assignment $s$ in which every node in $V_1$ chooses $a$ and every node in $V_2$ chooses $b$ is a $k$-equilibrium. Assume that there is a profitable deviation $s \betredge{K} s'$ such that $|K| \leq k$. Then all nodes in $K$ switch to $c$ and all nodes that choose $c$ in $s'$ are in $K$. So for all $v \in K$, $p_v(s') = |N_v \cap K|.$ So there is $v \in V_1 \cap K$ because otherwise the payoff of all nodes in $K$ would remain 0. But then $p_v(s') = |N_v \cap K| \leq \alpha(k-1) = \alpha p_v(s)$, which yields a contradiction. Note that $\SW(s) = k(k-1).$

It follows that the $\alpha$-approximate $k$-price of anarchy is unbounded if $k=1$ and at least
\[
\frac{\SW(\sigma)}{\SW(s)} = \frac{2 \alpha ((n-1) - (k-1)) + k-1}{k-1} = 
2 \alpha \left(\frac{n-1}{k-1} -1 \right) + 1
\]
for $k > 1.$
\qed
\end{proof}

\section{Missing proofs of Section \ref{sec:complexity}}

A crucial insight to prove the first part of the theorem is that it is sufficient to consider deviations that are \emph{simple}: We call a deviation $s \betredge{K} s'$ \emph{simple} if the subgraph $G[K]$ induced by $K$ is connected and all nodes in $K$ deviate to the same color, i.e., $s' = (x_K, s_{-K})$ for some $x \in M$.

\begin{lemma}
\label{lem:simple_ext}
Let $s \betredge{K} s'$ be an $(\alpha, k)$-improving deviation from $s$. Then there is also a simple such deviation. 
\end{lemma}
\begin{proof}
Let $s \betredge{K} s'$ be an $\alpha$-improving deviation with $|K| \leq k$. Pick an arbitrary $v \in K$ and let $x = s'_v$. Let $L$ consist of those nodes $u \in K$ for which $s'_u = x$ and $u$ is reachable in $G[K]$ from $v.$ Let $s'' = (x_L, s_{-L})$.  For all nodes $u\in L$, 
\begin{align*}
p_u(s'') = w(\{u,v\} \in E \mid s''_v = x\}) \geq w(\{u,v\} \in E \mid s'_v = x\}) = p_u(s') > \alpha p_u(s),
\end{align*}
where for the first inequality we use the fact that all neighbors of $u$ playing $x$ in $s'$ also do so in $s''$ by choice of $L.$
\qed
\end{proof} 

\subsection*{Proof of Theorem \ref{thm:verification}}

\begin{proof}[Theorem \ref{thm:verification}, $k=O(1)$]
By Lemma \ref{lem:simple_ext}, it is enough to check for all colors $x$ and all coalitions $K \subseteq \{v \in N \mid x \in S_v\}$ of size at most $k$ whether the deviation to $x$ is $\alpha$-improving. For each coalition of size at most $k$, this takes time $O(mk)$ and we check $O(n^k)$ such coalitions. So we need $O(mn^{k} k)$ time in total.
\qed
\end{proof}

\begin{proof}[Theorem \ref{thm:verification}, $k = n$]
We show that the problem can be reduced to the following \textsc{MinDegree} problem: Given a graph $G = (V,E)$ together with non-negative numbers $d_v$ for all nodes $v \in V$, find the inclusionwise maximal $K \su V$ such that for every $v \in K$, $w(\{\{v,j\} \in E \mid j \in K\}) > d_v$. It is not hard to see that a simple greedy algorithm solves the \textsc{MinDegree} problem in polynomial time (see Theorem~\ref{thm:min_deg} below for details).

The idea now is to use this algorithm to find an inclusionwise maximal coalition $K$ such that  $s \betredge{K} (x_K, s_{-K})$ is $\alpha$-improving. To this aim, consider the graph $G' = (N',E')$ induced by the set of vertices that can choose $x$ but do not do so in $s$. Let
$$
d_v = \alpha  p_v(s) - w(\{\{v,j\} \in E \mid s_j = x\}) - q^v(x) 
$$ 
for $v \in N'$. 
Fix a coalition $K \su N'$. Then for each $v \in K$, jointly deviating to $(x_{K}, s_{-K})$ is profitable iff the weight of its incident edges in the induced subgraph $(K, \{\{v,w\} \in E \mid v,w \in K\})$  is more than $d_v$. Indeed, 
\begin{align*}
p_v(x_{K}, s_{-K}) &= q^v(x)+  w(\{\{v,j\} \in E \mid j \in K\})\\
	& \qquad + w(\{\{v,j\} \in E \mid  s_j = x\})\\
 	&> \alpha p_v(s)
\end{align*}
if and only if $ w(\{\{v,j\} \in E \mid j \in K\}) > d_v$. So the problem reduces to finding the inclusionwise maximal solution of \textsc{MinDegree}.

Now the claim immediately follows: Let a joint strategy $s$ be given. For all colors $x \in M$, we compute the inclusionwise maximal $K_x$ such that $s \betredge{K_x} (x_K, s_{-K})$ is $\alpha$-improving. If there is a $K_x$ such that $K_x \neq \emptyset$, $s$ is not an $\alpha$-approximate strong equilibrium. Otherwise, there is no simple $\alpha$-improving deviation. So $s$ is an $\alpha$-approximate strong equilibrium by Lemma \ref{lem:simple_ext}.
\qed
\end{proof}

\begin{theorem}
\label{thm:min_deg}
\textsc{MinDegree} can be solved efficiently. 
\end{theorem}
\begin{proof}
Note that the empty set satisfies the degree conditions, so a solution always exists. Furthermore, for $K, K' \su N$ that satisfy the degree conditions, $K \cup K'$ also satisfies them. So there is a unique inclusionwise maximal $K$.

Consider the following algorithm: 
\begin{enumerate}
\item Initialize $K = N$. 
\item For each $v \in K$, check if $w(\{\{v,j\} \mid j \in K\}) > d_v$ and, if not, remove $v$ from $K$.
\item If a node has been removed in the last iteration of Step 2, repeat Step 2. Otherwise stop and output $K.$
\end{enumerate}
In every iteration of Step 2 except the last one a node is removed. So this algorithm runs in time $O(n(n+m)).$

We now show correctness. Let $K$ be as in the final state of the algorithm. Then $w(\{\{v,j\} \mid j \in K\}) > d_v$ for all $v \in K$ because otherwise another node would have been removed. 

Now we show inclusionwise maximality. Suppose some $K' \su N$ satisfies $w(\{\{v,j\} \mid j \in K\}) > d_v.$ for all $v \in K'.$ We claim that at each step of the algorithm, $K' \su K$: $K' \su K$ is true at the beginning of the algorithm. Consider an iteration of Step 2 and suppose $K' \su K$ holds at the beginning. Suppose some $v \in K'$ is removed in this iteration, and pick $v$ to be the first node of $K'$ that is removed. Then at this state, $\dg_{G[K]}(v) \geq \dg_{G[K']}(v) > d_v$ and thus $v$ is not removed, a contradiction. This shows the claim.
\qed
\end{proof}

\begin{proof}[Theorem \ref{thm:verification}, $\alpha$ fixed] This proof uses the same idea as \cite[Theorem 9]{us:coord}.
It is easy to verify that the problem is in co-NP: a certificate of a NO-instance is an $\alpha$-improving deviation of a coalition of size at most $k$.

We show the hardness by reduction of the complement of \textsc{Clique}, which is a co-NP-complete problem. Let $(G,k)$ be an instance thereof, with $G = (V,E)$. For $v \in V$ let $S_v = \{x_v, y\}$, where the colors $x_v$ are pairwise distinct. Furthermore, for every node $v \in V$ we add a node $u_v$ and edges $\{v,u_v\}$ of weight $k-2$.  This additional node can only choose the color $x_v.$ We give the edges in $E$ weight $\alpha$. Let $s$ be the joint strategy in which every node $v \in V$ chooses $x_v$. We claim that this is an $\alpha$-approximate $k$-equilibrium if and only if $G$ has no clique of size $k.$

Suppose $G$ has a clique $K$ of size $k$. Then jointly deviating to $y$ yields to each node in $K$ a payoff of $\alpha(k-1)$, whereas every node has a payoff of $k-2$ in $s$. So this is an $\alpha$-improving deviation. For the other direction, suppose that there is a profitable deviation $s \betredge{K} s'$ by a coalition $K$ of size at most $k$. Then every node in $K$ deviates to $y$ and hence belongs to $V.$ Since every node in $K$ has a payoff of $k-2$ in $s$, $p_v(s') > \alpha(k-2)$ for all $v \in K.$ So $v$ is connected to at least $k-1$ nodes in $K$. Since $K$ is of size at most $k$, this implies that $K$ is a clique of size $k.$
\qed
\end{proof}

\subsection*{Proof of Theorem \ref{thm:decide-existence}}

\begin{proof}[Theorem \ref{thm:decide-existence}, $k \ge 2$]
The problems are in NP because the corresponding verification problems are in P, as shown in Theorem \ref{thm:verification}.

Consider an instance $(G,l)$ of \textsc{Minimum Maximal Matching}, and let $\m G$ be the weighted coordination game constructed in Theorem \ref{thm:computation_2eq}. We claim that every 2-equilibrium is also a strong equilibrium. Then $G$ has a maximal matching of size at most $k$ iff $\m G$ has a strong equilibrium iff $\m G$ has a $k$-equilibrium, showing NP-hardness of both problems.

Let $s$ be a 2-equilibrium of $\m G.$ Assume that there is a profitable deviation $s \betredge{K} s'$. By Lemma \ref{lem:simple_ext}, we can assume that $K$ is connected and all nodes in $K$ can choose a common color which is different from the ones they choose under $s$. But this means that $|K| = 2$ by construction of $\m G.$ So $s$ is not a 2-equilibrium, a contradiction.
\qed
\end{proof}

\subsection*{Proof of Theorem \ref{thm:SE_forest}}

\begin{proof}[Theorem \ref{thm:SE_forest}]
Assume w.l.o.g. that $G$ is connected, i.e., is a tree. 
We can furthermore assume that $\m G$ is without individual preferences by adding dummy nodes to `simulate' the players' individual preferences: We add a node $v_{x,i}$ for each color $x$ and original player $i$, which we connect to $i$ by an edge of weight $q^i(x)$, and thus get a new graph coordination game $\m G'$ without individual preferences. This preserves strong equilibria and the fact that $G$ is a tree.

Fix an ordering of $G$ with root $r$. For each node $v \in V$ let $\m C_v \su V$ denote the set of children of $v$ and let $\m P_v \in V$ denote the parent of $v$ (if $v \neq r$). Let $\m T_v$ denote the subtree of $G$ rooted at $v$.

Consider the sequential-move game $\m G$ if the players choose their strategies according to the ordering of $G$, starting at the root. Then this game has a subgame perfect equilibrium $s$ which can be computed in polynomial time by backwards induction. For a player $v$, denote by $s[v \to x]$ the joint strategy of the subtree $\m T_v$ implemented under $s$ if $v$ plays $x \in S_v.$ For $u \in \m T_v$, let $s_u[v \to x]$ denote the color $u$ chooses under $s[v \to x]$. For $y \in S_{\m P_v}$ and $\sigma \in S_{\m T_v}$, let $p_v(y, \sigma)$ be the payoff of $v$ if $\m P_v$ plays $y$ and nodes in $\m T_v$ play according to $\sigma$. Let $\bar s$ denote  the joint strategy that is implemented with respect to $s.$

Note that the condition of $s$ being a subgame perfect equilibrium can then be written as follows: For all players $v$ and colors $x \in S_{\m P_v}$, choosing $y = s_v[\m P_v \to x]$ maximizes $p_v(x, s[v \to y])$ among $y\in S_v.$

Assume for a contradiction that there is a profitable deviation $\bar s \betredge{K} s'$ in the original game $G$. By Lemma~\ref{lem:simple_ext}, we can assume this deviation to be simple, i.e., $K$ is connected and all players in $K$ deviate to the same color $x.$ Let $v$ be the root of $K$. \\

\noindent \emph{Claim 1.} For all $u \in \m T_v$, $s'_u = x$ implies $s_u[\m P_u \to x] = x$, i.e., if $u$ chooses $x$ in $s'$, then he in particular chooses it in $s$ if his parent does so.

\begin{proof}
For $u \notin K$, this is easy to see: if $\bar s_{\m P_u} = x$, the claim trivially holds because $s_u [\m P_u \to x] = \bar s_u = s'_u = x.$ If, on the other hand, $\bar s_{\m P_u} \neq x$, the reasoning goes as follows. The fact that $\bar s_u = s'_u = x$ implies that it is a best response for $u$ to play $x$ if $\m P_u$ plays $\bar s_{\m P_u}$. If $\m P_u$ plays $x$ instead of $\bar s_{\m P_u}$, then $x$ is the only strategy of $u$ that yields a better payoff than before in $s$. So $x$ is the only best response to $\m P_u$ playing $x.$

Now, we show the claim for $u \in K$. We can assume that the claim holds for all children of $u.$ Assume for a contradiction that $s_u[\m P_u \to x] = y \neq x.$

Because $s$ is a subgame perfect equilibrium,
\[
p_u(\bar s) = p_u(\bar s_{\m P_u}, s[u \to \bar s_u]) 
\geq p_u(\bar s_{\m P_u}, s[u \to y]).
\]
Since $x \neq y$,
\[
p_u(\bar s_{\m P_u}, s[u \to y]) \geq p_u(x,  s[u \to y])
\]
Next, since $y$ is a best response to $x$,
\begin{align*}
p_u(x, s[u \to y]) \geq p_u(x, s[u \to x])
\end{align*}

Since the claim holds for all children of $u$, all children of $u$ choosing $x$ in $s'$ also do so in $s[u \to x]$ and thus contribute to the utility of $u.$ So
\begin{align*}
p_u(x, s[u \to x]) \geq p_u(x, s'_{\m T_u}).
\end{align*}
Because $s'_u = x$, we have
\[
p_u(x, s'_{\m T_u}) \geq p_u(s').
\]

Putting together the sequence of inequalities, we get $p_u(\bar s) \geq p_u(s')$, a contradiction to $u \in K.$
\qed
\end{proof}

%

\noindent \emph{Claim 2.} For the root $v$ of $K$ it holds that $p_v(s') \leq p_v(\bar s).$

\begin{proof}
Because $s'_v = x$, it follows from Claim 1 that all children of $v$ choosing $x$ under $s$ also do so in $s[v \to x].$ Furthermore, the strategy of $\m P_v$ is the same in $\bar s$ and $s'$ because $v$ is the root of $K$. Thus 
\[
p_v(s') = p_v(\bar s_{\m P_v}, s'_{\m T_v}) \leq p_v(\bar s_{\m P_v}, s[v \to x]).
\]
But 
\[
p_v(\bar s_{\m P_v}, s[v \to x]) \leq p_v(\bar s_{\m P_v}, s[v \to \bar s_v])= p_v(\bar s)
\]
because $s$ is a subgame perfect equilibrium. Together this implies $p_v(s') \leq p_v(\bar s)$. 
\qed
\end{proof}
This is a contradiction: $v$ should have profited from the deviation to $s'$. Thus $\bar s$ is indeed a strong equilibrium.
\qed
\end{proof}

The following example shows that in polymatrix coordination games this idea does not work: The joint strategy implemented in a subgame perfect equilibrium is not even necessarily a Nash equilibrium.
\begin{example}
\label{ex:polymatrix_tree}
We consider a polymatrix coordination game with individual preferences. The result then extends to polymatrix coordination games without individual preferences by introducing dummy nodes.

Consider the polymatrix coordination game w.i.p. on two nodes $u$ and $v$, both of which can choose two strategies `coordinate' ($c_u$ and $c_v$) and `play selfishly' ($s_u$ and $s_v$). The payoffs $q^{uv}$ are as follows: If both players coordinate, they get $4$; if $v$ plays selfishly, they get $0$ no matter what $u$ does; if $u$ plays selfishly and $v$ coordinates, they get $2$. Furthermore, both players have a preference for playing selfishly, namely $q^{u}(s_u) = 3$, $q^{u}(c_u) = 0$ and similarly $q^v(s_v) = 3, q^v(c_v)=0$.

Now, consider the sequential move game in which $u$ moves first. If $u$ coordinates, it is a best response for $v$ to coordinate, resulting in a payoff of $4$ for both. If $u$ plays selfishly, $v$ responds by playing selfishly too, resulting in a payoff of $3$ for both. So $u$ plays $c_u$ in the unique subgame perfect equilibrium, and the implemented strategies are $(c_u, c_v).$ But this is not a Nash equilibrium because $u$ can profitably defect by playing $s_u$ and getting a utility of $q^u(s_u) + q^{uv}(s_u, c_v) = 5.$
\qed
\end{example}

\section{Missing proofs of Section~\ref{sec:coord-mech}}
\label{sec:app:coord}

In our games the common payoff $q^{ij}$ of the bimatrix game on edge $\set{i, j}\in E$ is distributed equally to both $i$ and $j$. An idea that arises is to use different \emph{payoff sharing rules} to reduce the inefficiency. More specifically, suppose we fix for each edge $\set{i, j}$ shares $q_i^{ij}$ and $q_j^{ij}$ (summing up to $2q^{ij}$) that players $i$ and $j$ obtain, respectively, when playing the bimatrix game on $\set{i, j}$.
The following example shows that asymmetric payoff sharing rules are too weak to reduce the inefficiency, even in graph coordination games. 
\begin{example}\label{ex:payoff-sharing}
We review the example given in the Introduction in the context of payoff sharing rules. Recall that in this game we are given an arbitrary graph with unit edge weights and every player $i \in N$ can choose between a private color $c_i$ and a common color $c$. Now, no matter how the weight of 1 of each edge is shared among its endpoints, the joint strategy $s = (c_i)_{i \in N}$ is a Nash equilibrium. Thus the price of anarchy remains unbounded. 
\qed
\end{example}

\begin{proof}[Theorem~\ref{thm:strat-imp}]
For the first claim: 
We show this by reduction from \textsc{Minimum Vertex Cover.} Let an instance of this problem be given, i.e., a graph $G=(V,E)$ and a number $k.$ We construct a new graph $G'$ as follows. We start with $G$ and divide each edge $e = \{u,v\} \in E$ into two edges $\{u, v_e\}$ and $\{v_e, v\}$ (where $v_e$ is a new node corresponding to $e$). The set of colors is $M = \{c_v \mid v \in V\} \cup \{p_v \mid v \in V\} \cup \{p\}$. For every node $v \in V$, set $S_v = \{p_v, c_v\}$, and for each edge $e = \{u,v\} \in E$ set $S_{v_e} = \{c_u, c_v, p\}.$ Intuitively, the $c$-colors are common colors that favor coordination and the $p$-colors are private colors. Call the constructed coordination game $\m G.$

\begin{figure}[t]
\begin{center}
  \begin{tikzpicture}[scale = 1.5]
    \node (u) at (0, 0) {$u$};
    \node (v) at (2, 0) {$v$};
	\node (e) at (1, -.2) {};
	\node (e') at (1, -1.3) {};
    
    \node [label=below:{\lf $\{c_u, p_u\}$}] (u') at (0, -1.5) {$u$};
    \node [label=below:{\lf $\{c_u, c_v, p\}$}] (v_e) at (1, -1.5) {$v_e$};
    \node [label=below:{\lf $\{c_u, p_v\}$}] (v') at (2, -1.5) {$v$};

    \draw (u) -- node[auto]{$e$} (v);                                                       
    \draw (u') -- (v_e) -- (v');
    
    \draw [dashed, ->] (e) -- (e');
\end{tikzpicture}
\caption{Dividing an edge.}\label{fig:weighted-non-existence}
\end{center}
\end{figure}
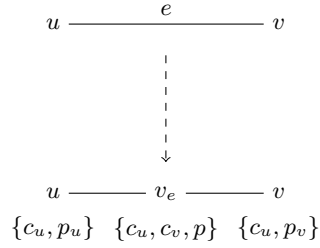
We claim that $G$ has a vertex cover of size $k$ if and only if there is a joint strategy $s_K$ of size $k$ that forces $2|E|$. Note that $2|E|$ is the maximum payoff in $\m G$ and is achieved iff all nodes $v_e$ have a payoff of 1.

Let a vertex cover $C$ of $G$ be given. We fix all nodes $v \in C$ to play $f_v := c_v.$ Let $s$ be a Nash equilibrium in $\m G[f_C].$ Then all nodes $v_e$ corresponding to edges in $G$ can achieve a payoff of $1$. So $\SW(s) \geq 2|E|.$

Conversely, suppose $f_K$ of size $k$ forces $2|E|$. We show that the set
\[
C := \{v \in V \mid \exists u \in K: \, f_K = c_u\}
\] 
with $|C| \leq k$ is a vertex cover of $G.$ Suppose that this is not the case. Then there exists an edge $e = \{u,v\} \in E$ such that no node in $V'$ is forced to choose $c_u$ or $c_v$. We construct a bad Nash equilibrium of $\m G[f_K]$ by starting at a joint strategy $s$ and playing best responses. Let all nodes in $V \sm K$ choose their private color in $s$. Then at no point it is beneficial for any node to switch to $c_u$ or $c_v$ because no other node plays either of these colors. So the payoff of $v_e$ in the Nash equilibrium is 0 and the Nash equilibrium achieves a social welfare of less than $2|E|$, a contradiction.

For the last two claims: 
We reuse the construction above. Let $s^*$ be such that $s^*_v = c_v$ for all $v \in V$ and $v_e$ chooses either $c_u$ or $c_v$ for $e = \{u,v\}.$ Note that $s^*$ is optimal with social welfare $2|E|$ and in the construction of $f_K$ in the previous theorem we only forced nodes to choose $s^*$. So the \textsc{Minimum Vertex Cover} instance admits a vertex cover of size $k$ iff the solution of Problem 2 is at most $k$ iff $\m G$ iff the solution of Problem 3 is at most $k$. This shows that \textsc{Minimum Vertex Cover} is reducible to either of these problems.
\qed
\end{proof}

\begin{proof}[Theorem~\ref{thm:forcing_sw}]
Let $\m G$ be a strategic game and $s'$ a joint strategy. Then $\m G$ is called \emph{$(\lambda, \mu)$-smooth with respect to $s'$} \cite{rough:rpoa} if for all joint strategies $s$,
\[
\sum_{v \in V} p_v(s'_v, s_{-v}) \geq \lambda \SW(s') + \mu \SW(s).
\]
It was shown\footnote{To be precise, \cite{rough:rpoa} assumes $s'$ to be arbitrary, but his proof ideas extend to the result stated above.} in \cite{rough:rpoa} that if $\m G$ is $(\lambda, \mu)$-smooth w.r.t. $s'$, then the social welfare of all coarse correlated equilibria is at least $\frac{\lambda}{1-\mu} \SW(s').$

Let $K$ consist of the $k$ players with the highest payoff in $s'.$ We show that $\m G[s'_K]$ is $(\frac{k}{n},0)$-smooth with respect to $s'$.

Let $s$ be an arbitrary strategy profile in $\m G[s'_K]$. For all $v \in V$,
\[
p_v(s'_v, s_{-v}) \geq q^v(s') + \sum_{j \in K, j \neq v} q^{vj}(s').
\]

It follows that
\begin{align*}
\sum_{v \in V}p_v(s'_v, s_{-v}) &\geq \sum_{v \in N} q^v(s') + \sum_{v \in N} \sum_{j \in K, j \neq v} q^{vj} \\
&\geq \sum_{j \in K} \Big(q^j(s') + \sum_{v \in N, v \neq j} q^{vj} \Big) = \sum_{k \in K} p_v(s').
\end{align*}
Because we chose $K$ to consist of the $k$ nodes with highest payoff w.r.t. $s'$, this is at least $\frac{k}{n} \SW(s').$ So $\m G[s'_K]$ is indeed $(\frac{k}{n}, 0)$-smooth w.r.t. $s'.$

Now, we show tightness in the above stated sense. Let $k$ be given. We construct an unweighted graph coordination game on the complete graph on $n = 2k+1$ nodes. Every node can choose the colors $a$ and $b.$ Let $s^*$ be the social optimum in which every player plays $a.$ Suppose we force a coalition $K$ of size $k$ to play $f_K.$ Let the rest of the nodes play in such a way that $k$ nodes choose $a$ and the rest chooses $b.$ It is easy to see that this is a Nash equilibrium of social welfare $k(k-1) + (k+1)k = 2k^2$. The optimal social welfare is $\SW(s^*) = 2nk$. The claim follows.
\qed
\end{proof}

\end{document}